\documentclass{article}

\def\noheaderplainsetup{

\topmargin=0pt \headheight=0pt \headsep=0pt  \oddsidemargin=0pt \evensidemargin=0pt  \textheight=9.1truein \textwidth=6.5truein}   

\noheaderplainsetup

\usepackage{amsfonts}    

\begin{document}

\newcommand{\lll}{\mbox{{\bf CL16}}}
\newcommand{\mmm}{\mbox{{\bf CL18}}}
\newcommand{\Bigmlc}{\mbox{{\Large $\wedge$}}}
\newcommand{\Bigmld}{\mbox{{\Large $\vee$}}}
\newcommand{\bigmlc}{\mbox{{\large $\wedge$}}}
\newcommand{\bigmld}{\mbox{{\large $\vee$}}}
\newcommand{\bigleftbrace}{\mbox{{\large $\{$}}}
\newcommand{\bigrightbrace}{\mbox{{\large $\}$}}}
\newcommand{\Bigleftbrace}{\mbox{{\Large $\{$}}}
\newcommand{\Bigrightbrace}{\mbox{{\Large $\}$}}}
\newcommand{\emptyrun}{\langle\rangle}
\newcommand{\legal}[2]{\mbox{\bf Lr}^{#1}_{#2}} 
\newcommand{\win}[2]{\mbox{\bf Wn}^{#1}_{#2}} 
\newcommand{\seq}[1]{\langle #1 \rangle}           


\newcommand{\ade}{\mbox{\large $\sqcup$}}      
\newcommand{\ada}{\mbox{\large $\sqcap$}}      
\newcommand{\gneg}{\neg}                  
\newcommand{\mli}{\rightarrow}                     
\newcommand{\mld}{\vee}    
\newcommand{\mlc}{\wedge}  
\newcommand{\add}{\hspace{0pt}\sqcup}           
\newcommand{\adc}{\hspace{0pt}\sqcap}
\newcommand{\st}{\mbox{\raisebox{-0.05cm}{$\circ$}\hspace{-0.13cm}\raisebox{0.16cm}{\tiny $\mid$}\hspace{2pt}}}
\newcommand{\cost}{\mbox{\raisebox{0.12cm}{$\circ$}\hspace{-0.13cm}\raisebox{0.02cm}{\tiny $\mid$}\hspace{2pt}}}


\newtheorem{theoremm}{Theorem}[section]
\newtheorem{factt}[theoremm]{Fact}
\newtheorem{definitionn}[theoremm]{Definition}
\newtheorem{lemmaa}[theoremm]{Lemma}
\newtheorem{conventionn}[theoremm]{Convention}
\newtheorem{claimm}[theoremm]{Claim}
\newtheorem{corollaryy}[theoremm]{Corollary}
\newtheorem{examplee}[theoremm]{Example}
\newtheorem{remarkk}[theoremm]{Remark}

\newenvironment{definition}{\begin{definitionn} \em}{ \end{definitionn}}
\newenvironment{theorem}{\begin{theoremm}}{\end{theoremm}}
\newenvironment{lemma}{\begin{lemmaa}}{\end{lemmaa}}
\newenvironment{fact}{\begin{factt}}{\end{factt}}
\newenvironment{corollary}{\begin{corollaryy}}{\end{corollaryy}}
\newenvironment{claim}{\begin{claimm}}{\end{claimm}}
\newenvironment{convention}{\begin{conventionn} \em}{\end{conventionn}}
\newenvironment{proof}{ {\bf Proof.} }{\  $\Box$ \vspace{.1in} }
\newenvironment{example}{\begin{examplee} \em}{\end{examplee}}
\newenvironment{remark}{\begin{remarkk} \em}{\end{remarkk}}

\title{Elementary-base cirquent calculus I: Parallel and choice connectives}
\author{Giorgi Japaridze
  \\  
 \\ Department of Computing Sciences, Villanova University, USA\\
 Email: giorgi.japaridze@villanova.edu\\
 URL: http://www.csc.villanova.edu/$^\sim$japaridz/
}
\date{}
\maketitle

\begin{abstract} {\em Cirquent calculus} is a proof system manipulating circuit-style constructs rather than formulas. 
Using it, this article constructs a sound and complete axiomatization $\lll$ of the propositional fragment of {\em computability logic} (the game-semantically conceived logic of computational problems) whose logical vocabulary consists of negation and parallel and choice connectives, and whose atoms represent elementary, i.e. moveless, games. 
\end{abstract}


\  


\section{Introduction}\label{intr}

{\em Computability logic}, or {\em CoL} for short, is a long-term project for developing a logic capable of acting as a comprehensive formal theory of computability in the same sense as classical logic is a formal theory of truth (see \cite{Japfin} for a survey). The approach starts by asking what kinds of mathematical objects ``computational problems''  are in their full generality, and finds that they can be most adequately understood as games played by a machine against its environment, with computability meaning existence of an (algorithmic) winning strategy for the machine. As its next step, CoL tries to identify a collection of the most natural, meaningful and potentially useful operations on games. These operations then form the connectives, quantifiers and other constructs of the logical vocabulary of CoL. Validity of a formula is understood as being ``always computable'', i.e. computable in virtue of the meanings of its logical operators regardless of how the non-logical atoms are interpreted. The final and most challenging step in developing CoL is finding sound and complete axiomatizations for ever more expressive fragments of this semantically construed logic. The present contribution adds one more brick to this edifice under construction.   

Among the main connectives of the language of CoL are {\em negation} (``{\em not}'') $\neg$, {\em parallel conjunction} (``{\em pand}'') $\mlc$, {\em parallel disjunction} (``{\em por}'') $\mld$, {\em choice conjunction} (``{\em chand}'') $\adc$, and {\em choice disjunction} (``{\em chor}'') $\add$. Where $G,H$ are games, the game-semantical meanings of the above connectives can be briefly characterized as follows. The game $\gneg G$ is nothing but $G$ with the roles  of the two players interchanged. $G\mlc H$ is a game playing which means playing $G$ and $H$ in parallel, where the machine wins if it wins in both components. $G\mld H$ differs from $G\mlc H$ only in that here winning in just one of the components is sufficient. $G\adc H$ is the game where, at the beginning, the environment chooses one of the two components, after which the game continues according to the rules of the chosen component. $G\add H$ is similar, only here it is the machine who makes an initial left-or-right choice. Game operations with similar intuitive characterizations have been studied by Lorenzen \cite{Lor61}, Hintikka \cite{Hintikka73} and Blass \cite{Bla72,Bla92} in their dialogue/game semantics, with Blass \cite{Bla92} being the first to systematically differentiate between the parallel and choice sorts of operations and pointing out their resemblance with the multiplicative ($\mlc,\mld$) and additive ($\adc,\add$) connectives of Girard's \cite{Gir87} linear logic. Many other operators of CoL have no known analogs in the literature. CoL also has two sorts of atoms: {\em general} atoms stranding for any games, and {\em elementary} atoms standing for propositions. The latter are understood as games with no moves, automatically won by the machine when true and lost when false. The fragments of CoL with only general atoms \cite{BauerLMCS, Cirq, Propint, Japfour, separating, taming1, taming2, Ver, qu, Xure1, Xure2} are called {\em general-base}, the fragments with only elementary atoms \cite{Japtocl1,Japtcs,lbcs,cl12} are called {\em elementary-base}, and the fragments where both sorts of atoms are present \cite{Japtocl2, Japtcs2, Japseq, Japtoggling, XuIGPL} are called {\em mixed-base}.  

All attempts to axiomatize the (whatever-base) full $\{\neg,\mlc,\mld,\adc,\add\}$-fragment  of CoL within the framework of traditional proof calculi had failed, and it was conjectured \cite{Bla92,Cirq} that such an axiomatization was impossible to achieve in principle even for the $\{\neg,\mlc,\mld\}$-subfragment. The recent work \cite{Anupam} by Das and Strassburger has positively verified this conjecture.   As a way to break the ice, \cite{Cirq} introduced the new sort of a proof calculus called {\em cirquent calculus}, in which a sound and complete axiomatization of the general-base $\{\neg,\mlc,\mld\}$-fragment of CoL was constructed; this result was later lifted to the mixed-base level in \cite{XuIGPL}.    Rather than being limited to tree-like objects such as formulas, sequents, hypersequents \cite{Avron} or deep-inference structures \cite{Gug07}, cirquent calculus deals with circuit-style constructs dubbed {\em cirquents}. Cirquents come in a variety of forms and sometimes, as in the present work or in \cite{XuIf,XuLast}, they are written textually rather than graphically, but their essence and main distinguishing feature remains the same: these are syntactic constructs explicitly allowing {\em sharing} of components between different subcomponents. Ordinary formulas of CoL are nothing but special cases of cirquents --- they are degenerate cirquents where nothing is shared. 

Sharing, itself, also takes different forms, such as two $\mld$-gates sharing a child, or two $\add$-gates sharing the left-or-right choice associated with them without otherwise sharing descendants.  Most cirquent calculus systems studied so far \cite{BauerLMCS, Cirq, Japdeep, taming1, taming2, XuIGPL} only incorporate the first sort of sharing. The idea of the second sort of sharing, dubbed {\em clustering}, was introduced and motivated in \cite{lmcs}. Among the potential benefits of it outlined in \cite{lmcs} was offering new perspectives on independence-free logic \cite{HS97}. Later work by Wenyan Xu \cite{XuIf,XuLast} made a significant progress towards materializing such a potential. The present work materializes another benefit offered by clustering: it constructs a sound and complete cirquent calculus axiomatization $\lll$ of the full elementary-base $\{\neg,\mlc,\mld,\adc,\add\}$-fragment of CoL.   No axiomatizations of any $\adc,\add$-containing fragments of CoL had been known so far (other than the brute-force  constructions of \cite{Japtocl1,Japtocl2,Japtcs,Japtcs2,Japseq,Japtoggling}, with their deduction mechanisms more resembling games than logical calculi). Generalizing from formulas to cirquents with clustering thus offers not only greater expressiveness, but also makes the otherwise unaxiomatizable CoL or certain fragments of it  amenable to being tamed as logical calculi.

\section{Games and strategies}\label{sgames}
As  noted, CoL understands computational problems as games played between two players, called  the {\em machine} and   the  {\em environment}. The symbolic names for these players are $\top$ and $\bot$, respectively. $\top$ is a deterministic mechanical device only capable of following algorithmic strategies, whereas there are no restrictions on the behavior of $\bot$.  Our sympathies are with $\top$, 
and by just saying ``won'' or ``lost'' without specifying a player, we always mean won or lost by $\top$. $\wp$ is always  a variable ranging over $\{\top,\bot\}$. 
\(\gneg \wp\) means $\wp$'s adversary, i.e. the player that is not $\wp$.

A {\bf move} is a finite string over the standard keyboard alphabet. 
A {\bf labeled move} is a move prefixed with $\top$ or $\bot$, with such a prefix ({\bf label}) indicating which player has made the move. 
A {\bf run} is a (finite or infinite) sequence of labeled moves, and a 
{\bf position} is a finite run.
Runs will be often delimited by ``$\langle$'' and ``$\rangle$'', with $\emptyrun$ thus denoting the {\bf empty run}.

\begin{definition}\label{game}
 A {\bf game}\footnote{In CoL, the proper name of the concept defined here is ``constant game'', with the word ``game'' reserved for a more general concept; however, since constant games are the only kinds of games we care about in the present paper, we omit the word ``constant'' and just say ``game''.}  is a pair $A=(\legal{A}{},\win{A}{})$, where:\vspace{10pt}

1. $\legal{A}{}$ is a set of runs satisfying the condition that a finite or infinite run is in $\legal{A}{}$ iff all of its nonempty finite --- not necessarily proper --- initial
segments are in $\legal{A}{}$ (notice that this implies $\emptyrun\in\legal{A}{}$). The elements of $\legal{A}{}$ are
said to be {\bf legal runs} of $A$, and all other runs are said to be {\bf illegal}. We say that $\alpha$ is a {\bf legal move} for a player $\wp$ in a position $\Phi$ of $A$ iff $\seq{\Phi,\wp\alpha}\in\legal{A}{}$; otherwise 
$\alpha$ is an {\bf illegal move}. When the last move of the shortest illegal initial segment of $\Gamma$  is $\wp$-labeled, we say that $\Gamma$ is a {\bf $\wp$-illegal run} of $A$; {\bf $\wp$-legal} means ```not $\wp$-illegal''. \vspace{5pt} 

2. $\win{A}{}$ is a function that sends every run $\Gamma$ to one of the players $\top$ or $\bot$, satisfying the condition that if $\Gamma$ is a $\wp$-illegal run of $A$, then $\win{A}{}\seq{\Gamma}=\gneg\wp$.\footnote{We write $\win{A}{}\seq{\Gamma}$ for $\win{A}{}(\Gamma)$.} When $\win{A}{}\seq{\Gamma}=\wp$, we say that $\Gamma$ is a {\bf $\wp$-won} (or {\bf won by $\wp$}) run of $A$; otherwise $\Gamma$ is {\bf lost by $\wp$}. Thus, an illegal run is always lost by the player who has made the first illegal move in it.  
\end{definition}

It is clear from the above definition that, when defining a particular game $A$, it would be sufficient to specify what {\em positions}  (finite runs) are legal, and what {\em legal runs} are won. Such a definition will then uniquely extend to all --- including infinite and illegal --- runs. We will implicitly rely on this observation in the sequel. 

A game is said to be {\bf elementary} iff it has no legal runs other than the (always legal) empty run $\emptyrun$. That is, an elementary game is a ``game'' without any (legal) moves,  automatically won or lost. There are exactly two  such games, for which we use the same symbols $\top$ and $\bot$ as for the two players: the game $\top$ automatically won by player $\top$, and the game $\bot$ automatically won by player $\bot$.\footnote{Precisely, we have $\win{\top}{}\emptyrun=\top$ and $\win{\bot}{}\emptyrun=\bot$.} Computability logic is a conservative extension of classical logic, understanding classical propositions as elementary games. And, just like classical logic, it sees no difference between any two true propositions such as ``$0= 0$'' and ``{\em Snow is white}'', and identifies them with the elementary game $\top$; similarly, it treats false propositions such as ``$0= 1$'' or ``{\em Snow is black}'' as the elementary game $\bot$. 

An {\bf HPM} (``Hard-Play Machine'') is a Turing machine with the additional capability of making moves. The adversary can also move at any time, with such moves being the only nondeterministic events from the machine's perspective. Along with the ordinary  read/write {\em work tape},\footnote{In computational-complexity-sensitive treatments, an HPM is allowed to have any (fixed) number of work tapes.} the machine also has an additional  tape called  the  {\em run tape}. The latter, at any time,  spells the ``current position'' of the play. The role of this tape is to make the interaction history fully visible to the machine.  It is read-only, and its content is automatically updated every time either player makes a move.

In these terms,  a  {\bf solution} ($\top$'s winning strategy) for a given  game $A$ is understood as an HPM $\cal M$ such that,  no matter how the environment acts during its interaction with $\cal M$ (what moves it makes and when), the run incrementally spelled on the run tape is a $\top$-won run of $A$. When this is the case, we write ${\cal M}\models A$ and say that ${\cal M}$ {\bf wins}, or {\bf solves}, $A$, and that $A$ is a {\bf computable} game.   

There is no need to define $\bot$'s strategies, because all possible behaviors by $\bot$ are accounted for by the different possible nondeterministic updates of the run tape of an HPM. 

In the above outline, we described HPMs in a relaxed fashion, without being specific about technical details such as, say, how, exactly, moves are made by the machine, how many moves either player can make at once, what happens if both players attempt to move ``simultaneously'', etc. As it turns out, all reasonable design choices yield the same class of winnable games as long as we consider a certain natural subclass of games called {\em static}. 
Intuitively, these are games where the relative speeds of the players are irrelevant because, as Blass has once put it, ``it never hurts a player to postpone making moves''. Below comes a formal definition of this concept.

For either player $\wp$, we say that a run $\Upsilon$ is a {\bf $\wp$-delay} of a run $\Gamma$ iff:\vspace{-5pt}
\begin{itemize}
\item for both players $\wp'\in\{\top,\bot\}$, the subsequence of $\wp'$-labeled moves of $\Upsilon$ is the same as that of $\Gamma$, and
\item for any $n,k\geq 1$, if the $n$th $\wp$-labeled move is made later than (is to the right of) the $k$th $\gneg\wp$-labeled move in $\Gamma$, then so is it in $\Upsilon$.\vspace{-5pt}
\end{itemize}
\noindent The above conditions mean that in  $\Upsilon$  each player has made the same sequence of moves as in $\Gamma$, only, in $\Upsilon$, $\wp$ might have been acting with some delay.

Now, we say that a game  $A$ is {\bf static} iff, whenever a run $\Upsilon$ is a $\wp$-delay of 
a run $\Gamma$, we have:\vspace{-2pt}
\begin{itemize}
\item if $\Gamma$ is a $\wp$-legal run of $A$, then so is $\Upsilon$;
\item if $\Gamma$ is a $\wp$-won run of $A$, then so is $\Upsilon$.\vspace{-2pt}
\end{itemize}

All games that we shall see in this paper are static. In fact, they  are not merely static, but belong to a special subclass of static games called ``enumeration games'',  where even the order in which the players make their moves is irrelevant, and thus runs can be seen as multisets rather than sequences of labeled moves. Precisely, an {\bf enumeration game} is a game $A$ such that, for any run $\Gamma$ and any permutation $\Delta$ of $\Gamma$, $\Gamma$ is a legal (resp. won) run of $A$ iff so is $\Delta$. 

Dealing only with static games, which makes timing technicalities fully irrelevant, allows us to describe and analyze strategies (HPMs) in a relaxed fashion. For instance, imagine HPM $\cal N$ works by simulating and mimicking the work and actions of another HPM $\cal M$ in the scenario where $\cal M$'s imaginary adversary acts in the same way as $\cal N$'s own adversary. Due to the simulation overhead, $\cal N$ will generally be much slower than $\cal M$ in responding to its adversary's moves.  Yet, we may safely assume/pretend that the speeds of the two machines do not differ   and thus they will be generating identical runs.  This is ``even more so'' when we deal with enumeration games. In what follows we will often implicitly rely on this observation. 



\section{Syntax}\label{Syntax}

We fix an infinite list 
of syntactic objects called {\bf elementary game letters}, for which we will be using 
$p,q,r$ as metavariables. A {\bf positive} (resp. {\bf negative}) {\bf literal} is the expression $p$ (resp. $\neg p$), where $p$ is an elementary game letter. Here $p$ is said to be the {\bf type} of the literal.

We further fix two pairwise disjoint infinite sets $\mathbb{C}(\add)$ and $\mathbb{C}(\adc)$ of decimal numerals. The elements of $\mathbb{C}(\add)\cup\mathbb{C}(\adc)$ are said to be {\bf clusters}. A cluster $c$ is said to be {\bf disjunctive} if $c\in \mathbb{C}(\add)$, and {\bf conjunctive} if $c\in\mathbb{C}(\adc)$.  

 The symbol $\mld$ (resp. $\mlc$) is said to be {\bf parallel disjunction} (resp. {\bf parallel conjunction}).   A {\bf choice disjunction} (resp. {\bf choice conjunction}) is a pair $\add^c$ (resp. $\adc^c$), where $c$ is a disjunctive (resp. conjunctive) cluster.  A common name for disjunctions and conjunctions of either sort is ``{\bf connective}'', and the corresponding symbol $\mld,\mlc,\add$ or $\adc$ is said to be the {\bf type} of the connective. Given a choice connective $\add^c$ or $\adc^c$, $c$  is said to be its {\bf cluster}; in this case we may as well say that the connective {\bf belongs to} --- or {\bf is in} --- cluster $c$. 

\begin{definition}\label{basecir}
A {\bf cirquent} is defined inductively as follows: 
\begin{itemize}
\item $\top$ and $\bot$ are cirquents. 
\item Each literal is a cirquent. 
\item If $A$ and $B$ are cirquents, then $(A)\mld(B)$ is a cirquent.
\item If $A$ and $B$ are cirquents, then $(A)\mlc(B)$ is a cirquent. 
\item If $A$ and $B$ are cirquents and $c$ is a conjunctive cluster,  then $(A)\adc^c(B)$ is a cirquent.
\item If $A$ and $B$ are cirquents and $c$ is a disjunctive cluster,  then  $(A)\add^c(B)$ is a cirquent. 
\end{itemize}   
\end{definition}

By a {\bf cluster of}  a cirquent $C$ we shall mean the cluster $c$ of some choice connective occurring in $C$. In such a case we may as well say that cluster $c$ {\bf occurs} in $C$.

When writing cirquents, parentheses will usually be omitted if this causes no ambiguity. When doing so, it is our convention that choice connectives take precedence over parallel connectives.  So, for instance, $A\adc^c B\mld C$  
means $(A\adc^c B)\mld C$ rather than $A\adc^c (B\mld C)$. 

Sometimes we may write an expression such as $A_1\mld \ldots\mld A_n$, where $n$ is a (possibly unspecified)  natural number with $n\geq 2$. This is to be understood as any (unspecified) order-respecting $\mld$-combination of the cirquents $A_1,\ldots,A_n$. ``Order-respecting'' in the sense that $A_1$ is the leftmost item of the combination, then comes $A_2$, then $A_3$, etc.   Similarly for 
$A_1\mlc\ldots\mlc A_n$.  So, for instance, both $(A\mlc B)\mlc C$ and $A\mlc(B\mlc C)$ --- and no other cirquent --- can be written as $A\mlc B\mlc C$. 

Officially, as we see, $\neg$ ({\bf negation}) is only allowed to be applied to elementary game letters. Shall we write $\neg E$ where $E$ is not an elementary game letter, it is to be understood as an abbreviation defined by: $\neg\neg A=A$; $\neg(A\mlc B)=\neg A\mld\neg B$; $\neg(A\mld B)=\neg A\mlc \neg B$; $\neg(A\adc^cB)=\neg A\add^c\neg B$; $\neg(A\add^c B)=\neg A\adc^c\neg B$. Similarly, $A\mli B$ is an abbreviation of $(\neg A)\mld B$. When writing cirquents, parentheses will usually be omitted if this causes no ambiguity. When doing so, it is our convention that $\neg$ has the highest precedence, then comes $\mli$, then come the choice connectives, and finally the parallel connectives. So, for instance, $\neg A\mld B\mli C\mlc D\adc^cE$ means $((\neg (A))\mld (B))\mli ((C)\mlc ((D)\adc^c(E)))$, i.e., 
$((A)\mlc (\neg(B)))\mld ((C)\mlc ((D)\adc^c(E)))$.

We define the {\bf root} of a cirquent $C$ to be $C$ itself if $C$ is $\top$, $\bot$ or a literal, and $\mld$ (resp. $\mlc$, resp. $\add^c$, resp. $\adc^c$) if $C$ is of the form $A\mld B$ (resp. $A\mlc B$, resp. $A\add^c B$, resp. $A\adc^c B$). When $r$ is the root of $C$, we say that $C$ is {\bf $r$-rooted}. 

\section{Semantics} 

We define {\em LegRuns} as the set of all runs satisfying 
the following conditions:\vspace{5pt}

1. Every move of $\Gamma$ is the string $c.0$ or $c.1$, where $c$ is a cluster. 

2. Whenever $\Gamma$ contains a move $c.i$ where $c$ is a disjunctive cluster, the move is $\top$-labeled. 

3.  Whenever $\Gamma$ contains a move $c.i$ where $c$ is a conjunctive cluster,  the move is $\bot$-labeled. 

4.  For any cluster $c$,  $\Gamma$ contains at most one move of the form $c.i$.\vspace{5pt}

The intuitive meaning of condition 1 is that every move signifies a choice ``left'' ($0$) or ``right'' ($1$) in some cluster; conditions 2 and 3 say that $\top$ moves (chooses) only in disjunctive clusters and $\bot$ only in conjunctive clusters; and condition 4 says that, in any given cluster, a choice can be made only once. 

Given a run $\Gamma\in LegRuns$, we say that a cirquent of the form $A\add^c B$ or $A\adc^c B$ is {\bf $\Gamma$-resolved} iff 
$\Gamma$ contains (exactly) one of the moves $c.0$ or $c.1$; then by the {\bf $\Gamma$-resolvent} of the cirquent we mean $A$ if such a move is  
$c.0$, and $B$ if it is $c.1$. ``$\Gamma$-unresolved'' means ``not $\Gamma$-resolved''. When $\Gamma$ is clear from the context, we may omit a reference to it and simply say ``resolved'', ``unresolved''  or ``resolvent''. 

An {\bf interpretation} is a function $^*$ which assigns to each elementary game letter $p$ an element $p^*$ of $\{\top, \bot\}$. Intuitively, such a function tells us whether $p$, as a proposition, is true or false. 

\begin{definition} \label{semantics}
Each  cirquent $C$ and interpretation $^*$ induces a unique game $C^*$, which we may refer to as ``$C$ {\bf under} the interpretation $^*$''. The 
set $\legal{C^*}{}$ of legal runs of such a game is nothing a but {\em LegRuns}. Since  $\legal{C^*}{}$ does not depend on $C$ or $^*$, subsequently we shall simply say ``legal run'' rather than  ``legal run of $C^*$''. 
The $\win{C^*}{}$ component of the game $C^*$ is defined by stipulating that a legal run  $\Gamma$ is a won (by the machine) run of $C$ iff one of the following conditions is satisfied: 

1.  $C$ is $\top$. 

2. $C$ is a positive (resp. negative) literal and, where $p$ is the type of that literal, $p^*=\top$ (resp. $p^*= \bot$).

3. $C$ is $A_0\mld A_1$ (resp. $A_0\mlc A_1$) and, for at least one (resp. both) $i\in\{0,1\}$,  $\Gamma$ is a won run of $A_i$.

4.  $C$ is $A_0\add^c A_1$, it is resolved and, where $A_i$ is the resolvent, $\Gamma$ is a won run of $A_i$. 

5. $C$ is $A_0\adc^c A_1$ and either it is unresolved, or else, where $A_i$ is the resolvent, $\Gamma$ is a won run of $A_i$. 
\end{definition}

\begin{definition} \label{krisha}
 Consider a cirquent $C$. 

1. For an interpretation $^*$, a {\bf solution} of $C$ under $^*$, or simply a solution of $C^*$, is an HPM $\cal H$ such that $\cal H\models C^*$. We say that $C$ is {\bf computable} under $^*$, or simply that $C^*$ is computable, iff $C^*$ has a solution. 

2. A {\bf logical} (or uniform) {\bf solution} of $C$ is an HPM $\cal H$ such that, for any interpretation $^*$, $\cal H$ is a solution of $C^*$. We say that $C$ is {\bf valid} iff it has a logical solution.\footnote{In CoL, this sort of validity is called {\bf logical} (or {\bf uniform}) {\bf validity}. There is also another natural sort of validity, called {\bf nonlogical} (or {\bf multiform}) {\bf validity}. Namely, a cirquent (or formula) $C$ is multiformly valid iff, for any interpretation $^*$, 
$C^*$ is computable. Nonlogical validity will not be considered in this paper.}
\end{definition}

\begin{remark}\label{cl1}
The cirquents in the present sense can be understood as generalizations of the formulas of system {\bf CL1}  of CoL constructed in \cite{Japtocl1}. Syntactically, the formulas differ from cirquents only in that no clusters are attached to $\add,\adc$. Each formula $F$ can be seen as a cirquent $C$ where no two different occurrences of a choice connective belong to the same cluster, i.e., as a cirquent with no sharing of choices associated with $\add,\adc$. More specifically, $C$ is a cirquent obtained from $F$ via superscripting each occurrence of $\add$ by a unique disjunctive cluster and each occurrence of $\adc$ by a unique conjunctive cluster. Let us call such a $C$ a {\em cirquentization} of $F$. We claim without a proof that, given a formula $F$ and a cirquentization $C$ of it, the two are semantically equivalent. Namely, any HPM $\cal F$ can be transformed into an HPM $\cal C$ --- and vice versa --- so that, for any interpretation $^*$, we have ${\cal F}\models F^*$ iff ${\cal C}\models C^*$ (with $F^*$ understood as in \cite{Japtocl1}). Consequently, $F$ is valid iff $C$ is so. 
\end{remark}

\section{Axiomatics}

By a {\bf rule of inference} we mean a set $\cal R$ of pairs $\vec{A}\leadsto B$, called {\bf applications} of $\cal R$,  where $\vec{A}$ 
is a tuple consisting of one or two cirquents, called the {\bf premise(s)}, and $B$ is a cirquent, called the {\bf conclusion}. When $\vec{A}\leadsto B$  
is in $\cal R$, we say that $B$ {\bf follows} from $\vec{A}$  by rule $\cal R$. 


In this section and later we will be using the notation $X[E_1,\ldots,E_n]$ to stand for a cirquent (intuitively  ``of {\bf structure} $X$'') together with some fixed subcirquents $E_1,\ldots,E_n$. Then, if we later write $X[F_1,\ldots,F_n]$ in the same context, it should be understood as the result of replacing, in $X[E_1,\ldots,E_n]$, all occurrences of $E_1,\ldots,E_n$ by $F_1,\ldots,F_n$, respectively. When this notation is used in the formulation of a rule of inference, our convention is that the context is always set by the conclusion. So, for instance, if we have a (sub)expression $X[E]$ in the conclusion and $X[F]$ in a premise, then $X[F]$ is the result of replacing all occurrences of $E$ by $F$ in $X[E]$ rather than vice versa. 


Below is a full list of the rules of inference of our system $\lll$.  The first seven rules come in two versions, between which we shall later differentiate by suffixing the name of the rule with ``(a)'' for the first version and  ``(b)''  for the second version. The last rule takes two premises, while all other rules take a single premise. The rules are written schematically, with $A,B,C,D$ (possibly with indices) acting as variables for subcirquents, $a,b,c$ as variables for clusters, and $X,Y$ as variables for ``structures''. The names of these rules have been chosen according to the conclusion-to-premises (rather than premises-to-conclusion) intuitions. 


\begin{description}
\item[Commutativity:]  $X[B\mld A]\leadsto X[A\mld B]$ \ and \ $X[B\mlc A]\leadsto X[A\mlc B]$. 
\item[Associativity:]  $X[A\mld(B\mld C)]\leadsto  X[(A\mld B)\mld C)]$ \ and \ $X[A\mlc(B\mlc C)]\leadsto  X[(A\mlc B)\mlc C)]$.
\item[Identity:]  $X[A]\leadsto X[A\mld \bot]$ and  $X[A]\leadsto  X[A\mlc \top]$.  
\item[Domination:]  $X[\top]\leadsto X[ A\mld \top]$ and $X[\bot]\leadsto  X[A\mlc \bot]$. 
\item[Choosing:]  $X[A_1,\ldots,A_n]\leadsto  X[A_1 \add^c B_1,\ldots,A_n \add^c B_n]$ \ and  \ $X[B_1,\ldots,B_n]\leadsto  X[A_1 \add^c B_1,\ldots,A_n \add^c B_n]$,  \ where $A_1\add^c B_1$, \ldots, $A_n \add^c B_n$ are all $\add^c$-rooted subcirquents of the conclusion. 
\item[Cleansing:] $X\bigl[Y[A]\adc^c C\bigr]\leadsto X\bigl[Y[A\adc^c B]\adc^c C\bigr]$ \ and \ $X\bigl[C \adc^c Y[B]\bigr]\leadsto X\bigl[C \adc^c Y[A\adc^c B]\bigr]$.
\item[Distribution:]     $X[(A\mld C)\mlc(B \mld C)]\leadsto  X[(A\mlc B)\mld C]$ \ and \  $X[(A\mld C)\adc^c (B \mld C)]\leadsto  X[(A\adc^c B)\mld C]$.
\item[Trivialization:] $X[\top]\leadsto  X[\neg p\mld p]$,  where $p$ is an elementary letter.
\item[Quadrilemma:] $X\Bigl[\Bigl(A\mlc (C \adc^b D)\bigr) \adc^a \bigl( B\mlc (C \adc^b D)\bigr)\Bigr)\adc^c \Bigl(\bigl(((A\adc^a B)\mlc C\bigr)
      \adc^b\bigl( (A\adc^a B)\mlc D\bigr)\Bigr)\Bigr] \leadsto  X[(A\adc^a B)\mlc (C \adc^b D)]$,  where $c$ does not occur in the conclusion. 
\item[Splitting:] $A,B\leadsto  A\adc^c B$,  where neither $A$ nor $B$ has an occurrence of $c$.
\end{description}

A {\bf proof} of a cirquent $A$ is a sequence $C_1,\ldots,C_n$ ($n\geq 1$) of cirquents such that $C_1 = \top, C_n = A$ and, for each $i\in \{2,\ldots,n\}$, $C_i$ follows by one of the rules of inference from some earlier cirquents in the sequence. Thus, $\top$ is the only axiom of $\lll$. 

\begin{example}\label{exm1}
Below is a proof of $p\mlc q\add^c r \mli (p\mlc q)\add^d (p\mlc r)$, i.e. of $(\neg p\mld \neg q\adc^c \neg r )\mld (p\mlc q)\add^d (p\mlc r)$. For brevity, consecutive applications of Commutativity or Associativity have been combined together in single steps.  

1. $\top$ \hspace{15pt} Axiom 

2. $\top  \mlc \top $ \hspace{15pt} Identity(b): 1

3. $ (\neg q\mld \top)\mlc ( \neg p\mld \top)$ \hspace{15pt} Domination(a): 2 (twice)

4. $ \bigl(\neg q\mld (\neg p\mld p)\bigr)\mlc \bigl( \neg p\mld (\neg q\mld q)\bigr)$ \hspace{15pt} Trivialization: 3 (twice)

5. $ \bigl((\neg q\mld \neg p)\mld p\bigr)\mlc \bigl( (\neg p\mld \neg q)\mld q\bigr)$ \hspace{15pt} Associativity(a): 4 (twice)

6. $ \bigl(p\mld (\neg q\mld \neg p)\bigr)\mlc \bigl(q\mld (\neg q\mld \neg p)\bigr)$ \hspace{15pt} Commutativity(a): 5 (three times)

7. $ (p\mlc q)\mld (\neg q\mld \neg p)$ \hspace{15pt} Distribution(a): 6

8. $(\neg q\mld \neg p) \mld (p\mlc q)$ \hspace{15pt} Commutativity: 7

9. $(\neg q\mld \neg p) \mld (p\mlc q)\add^d (p\mlc r)$ \hspace{15pt} Choosing(a): 8

10. $(\neg r\mld \top)\mlc (\neg p\mld \top) $ \hspace{15pt} Domination(a): 2 (twice)

11. $\bigl(\neg r\mld (\neg p \mld p)\bigr)\mlc \bigl(\neg p\mld (\neg r \mld r )\bigr) $ \hspace{15pt} Trivialization: 10 (twice)

12. $\bigl((\neg r\mld \neg p) \mld p\bigr)\mlc (\bigl(\neg p\mld \neg r) \mld r \bigr) $ Associativity(a): 11 (twice)

13. $\bigl(p\mld (\neg r\mld \neg p) \bigr)\mlc \bigl(r\mld (\neg p\mld \neg r)\bigr ) $ \hspace{15pt} Commutativity(a): 12 (twice)

14. $(p\mlc r)\mld (\neg r\mld \neg p)$ \hspace{15pt} Distribution(a): 13

15. $(\neg r\mld \neg p) \mld (p\mlc r)$ \hspace{15pt} Commutativity(a): 14

16. $(\neg r\mld \neg p) \mld (p\mlc q)\add^d (p\mlc r)$ \hspace{15pt} Choosing(b): 15

17. $\bigl((\neg q\mld \neg p) \mld (p\mlc q)\add^d (p\mlc r)\bigr)\adc^c \bigl((\neg r\mld \neg p) \mld (p\mlc q)\add^d (p\mlc r)\bigr)$ \hspace{15pt} Splitting: 9,16

18. $(\neg q\mld \neg p)\adc^c (\neg r\mld \neg p) \mld (p\mlc q)\add^d (p\mlc r)$ \hspace{15pt} Distribution(b): 17

19. $(\neg q\adc^c \neg r \mld\neg p )\mld (p\mlc q)\add^d (p\mlc r)$ \hspace{15pt} Distribution(b): 18

20. $(\neg p\mld \neg q\adc^c \neg r)\mld (p\mlc q)\add^d (p\mlc r)$ \hspace{15pt} Commutativity(a): 19
\end{example}

\section{The preservation lemma} 

\begin{lemma} \label{pres}
 Consider an arbitrary interpretation $^*$. 

1. Each application of any of the rules of $\lll$ preserves computability under $^*$ in the premises-to-conclusion direction, i.e., if all premises are computable under $^*$, then so is the conclusion.   
 
2. Each application of any of the rules of $\lll$ other than Choosing also preserves computability under $^*$ in the conclusion-to-premises direction, i.e., if the conclusion is computable under $^*$, then so are all premises. 
\end{lemma}

\begin{proof} Consider  an arbitrary interpretation  $^*$. Since $^*$ is going to be fixed throughout this proof, for readability  we agree to omit explicit references to it. So, for instance, where $E$ is a cirquent, we may write $E$ instead of $E^*$, or say ``\ldots solution of $E$'' instead of  ``\ldots solution of $E$ under $^*$''. Throughout this and some later proofs, when trying to show that a given machine $\cal H$ is a solution of a given game $G$, we implicitly rely on what is called the ``clean environment assumption''. According to it, $\cal H$'s environment never makes moves that are not legal moves of $G$. Assuming that this condition is satisfied is legitimate, because, if $\cal H$'s environment makes an illegal move, $\cal H$ automatically wins. 

If $E\leadsto F$ is an application of any of the rules other than Splitting or Choosing, it is not hard to see that  $E$  and $F$ are identical as games. So, a solution of $E$ is automatically a solution of $F$, and vice versa. Let us just look at Cleansing(a) as an illustrative example. Consider an application $X\bigl[Y[A]\adc^c C\bigr]\leadsto X\bigl[Y[A\adc^c B]\adc^c C\bigr]$ of this rule. Let $\Gamma$ be an arbitrary legal run. We want to show that $\Gamma$ is a won run of $E$ iff it is a won run of $F$. If $c$ is unresolved in $\Gamma$, then the $Y[A\adc^c B]\adc^c C$ component of the conclusion will be won just like the $Y[A]\adc^c C$ component of the premise. Since the two cirquents only differ in that one has $Y[A\adc^c B]\adc^c C$ where the other has $Y[A]\adc^c C$, we find that $\Gamma$ is a won run of both games or neither. Now assume $c$ is resolved, i.e., $\Gamma$ contains the move $c.i$ ($i=0$ or $i=1$). If $i=1$, then $\Gamma$ is a won run of  $X[Y[A\adc^c B]\adc^c C]$  iff it is a won run of $X[C]$  iff it is a won run of $X[Y[A]\adc^c C]$.  And if $i=0$, then $\Gamma$ is a won run of 
$ X\bigl[Y[A\adc^c B]\adc^c C\bigr]$ iff it is a won run of $ X\bigl[Y[A\adc^c B]\bigr]$ iff it is a won run of $ X\bigl[Y[A]\bigr]$ iff it is a won run of $ X\bigl[Y[A\adc^c C]\bigr]$. Thus, in either case, the conclusion is won iff so is the premise.  

Consider an application $A,B \leadsto  A\adc^c B$ of {\em Splitting}.  

For the premises-to-conclusion direction, assume the premises are computable, namely, HPMs ${\cal M}_A$  and ${\cal M}_B$  are solutions of $A$  and $B$, respectively. Let $\cal N$ be an HPM which, at the beginning of the play, waits till the environment makes one of the moves $c.0$ or $c.1$.
 After that, where $\alpha_1,\ldots,\alpha_n$ are the moves made by the 
environment before the move $c.0$ (resp. $c.1$) was made, $\cal N$ starts simulating ${\cal M}_A$ (resp. ${\cal M}_B$), with 
$\bot\alpha_1,\ldots,\bot\alpha_n$ 
on the imaginary run tape of the latter at the very first clock cycle. Whenever $\cal N$ sees that the simulated machine ${\cal M}_A$ (resp. ${\cal M}_B$) made a move, $\cal N$ makes the same move; $\cal N$ also periodically checks its own run tape to see if the 
environment has made any new moves in the real play and, if yes, it appends those ($\bot$-prefixed) moves to the imaginary run tape of the simulated machine. In more relaxed and intuitive terms, what we just said about the actions of $\cal N$ after the environment has moved $c.0$ (resp. $c.1$) can be put as ``$\cal N$ plays exactly like ${\cal M}_A$ (resp. ${\cal M}_B$) would play in the scenario where, at the very start of the play, the environment made the moves $\alpha_1,\ldots,\alpha_n$''. Later, in similar situations, we shall usually describe and analyze HPMs in relaxed terms, without going into technical details of simulation and without even using the word ``simulation''. Since we  exclusively deal  with static games, this relaxed approach is safe and valid (see the end of Section \ref{sgames}). Anyway, it is not hard to see that our $\cal N$ is a solution of $A\adc^c B$. 

For the conclusion-to-premises direction, assume $\cal N$ is a solution of $A\adc^c B$. Let ${\cal M}_A$ (resp. ${\cal M}_B$) be an HPM which plays just like $\cal N$ would in the scenario where, at the very start of the play, $\cal N$'s adversary made the move $c.0$ (resp. $c.1$). Obviously  ${\cal M}_A$ and ${\cal M}_B$   are solutions of $A$ and $B$, respectively.

Consider an application $X[A_1,\ldots,A_n]\leadsto  X[A_1 \add^c B_1,\ldots,A_n \add^c B_n]$ of {\em Choosing(a)}, and assume $\cal M$ is a 
solution of the premise. Let ${\cal N}$ be an HPM which, at the beginning of the game, makes the move $c.0$, after which it plays exactly as 
$\cal M$ would. Obviously $\cal M$ 
is a solution of the conclusion. {\em Choosing(b)} will be handled in a similar way. 
\end{proof}

The following is an immediate corollary of Lemma \ref{pres}:  

\begin{corollary}\label{prescor} 
1. Each application of any of the rules of $\lll$ preserves validity in the premise-to-conclusion direction, i.e., if all premises are valid, then so is the conclusion. 

2. Each application of any of the rules of $\lll$ other than Choosing also preserves validity in the conclusion-to-premise direction, i.e., if the conclusion is valid, then so are all premises.
\end{corollary}

\begin{remark}\label{rfeb5} Lemma \ref{pres} and Corollary \ref{prescor} state the existence of certain solutions. A look back at our proof of those statements reveals that, in fact, this existence is constructive. Namely, in the case of clause (a) of Lemma \ref{pres}, for any given rule, there is a $^*$-independent effective procedure which extracts an HPM $\cal M$ from the premise(s), the conclusion and HPMs that purportedly solve the premises under $^*$; as long as these purported solutions are indeed solutions, $\cal M$ is a solution of the conclusion under $^*$. Similarly for clause (b). In the case of clause (a)  of Corollary \ref{prescor},  for any given rule, there is an effective procedure which extracts an HPM $\cal M$ from the premise(s), the conclusion and purported logical solutions  of the premises; as long as these purported logical solutions are indeed logical solutions, $\cal M$ is a logical solution of the conclusion. Similarly for clause (b). 
\end{remark}

\section{Soundness and completeness}
Below we use the standard notation $^na$ (``tower of $a$'s of height $n$'') for tertration, defined inductively by $^1 a= a$ and $^{n+1}a=a^{( ^na)}$.  So, for instance, $^3 5=5^{5^5}$.
\begin{definition}\label{rankdef}
The {\bf rank} $\overline{C}$ of a cirquent $C$ is the number defined as follows:  

1. If $C$ is $\top$, $\bot$ or a literal, then $\overline{C}\ =\ 1$.

2. If $C$ is  $A\add^c B$ or $A\adc^c B$, then $\overline{C}\ =\ \overline{A}+\overline{B}$.

3. If $C$ is  $A\mlc B$, then $\overline{C}\ =\ 5^{\overline{A}+\overline{B}}$.

4. If $C$ is $A\mld B$, then  $\overline{C}\ = \ ^{\overline{A}+\overline{B}}5$.
\end{definition}

\begin{lemma}\label{monot}
The rank function is {\bf monotone} in the following sense. Consider a cirquent $A$ with a subcirquent $B$. Assume $B'$ is a cirquent with $\overline{B'}< \overline{B}$, and $A'$ is the result of replacing an occurrence of $B$ by $B'$ in $A$. Then $\overline{A'}< \overline{A}$.
\end{lemma}

\begin{proof} This is so due to the monotonicity of the functions $x+y$, $5^x$ and $^x5$. \end{proof}
 
A {\bf surface occurrence} of a subcirquent or a connective in a given cirquent is an occurrence which is not in the scope of a choice connective. 

\begin{definition}\label{pd}
 We say that a cirquent $D$ is {\bf pure} iff the following conditions are satisfied: 

1. $D$ has no surface occurrences of $\bot$ unless $D$ itself is $\bot$. 

2. $D$ has no surface occurrence of $\mlc$ which is in the scope of $\mld$. 

3. $D$ has no surface occurrence of $\adc^c$  (whatever cluster $c$) which is in the scope of $\mld$.
 

4. $D$ has no surface occurrence of the form $A_1\mld\ldots\mld A_n$ such that, for some elementary letter $p$, both 
$p$ and $\neg p$ are among $A_1,\ldots,A_n$. 

5. $D$ has no surface occurrences of $\top$ unless $D$ itself is $\top$. 

6. If $D$ is of the form $A_1\mlc\ldots\mlc A_n$ ($n\geq 2$), then at least one $A_i$ ($1\leq i\leq n$) is not of the form $B\adc^c C$.


7. If $D$ is of the form $A\adc^c B$, then neither $A$ nor $B$ contains the cluster $c$. 
\end{definition}

Below we describe a procedure which takes a cirquent $D$ and applies to it a series of modifications. Each modification changes the value of $D$ so that the old value of $D$ follows from the new value by one of the single-premise rules (other than Choosing) of $\lll$. The procedure is divided into 7 stages, and the purpose of each stage $i\in\{1,\ldots,7\}$ is to make $D$ satisfy the corresponding condition $i$ of Definition \ref{pd}.\vspace{5pt}   

{\bf Procedure Purification} applied to a cirquent $D$:  
Starting from Stage 1, each  of the following 7 stages is a loop that should be iterated until it no longer modifies (the current value of) $D$; then the procedure goes to the next stage, unless the current stage was Stage 7, in which case the procedure returns (the then-current value of) $D$ and terminates.  

{\em Stage 1}:   If $D$ has a surface occurrence of the form $\bot\mld A$ or  $A\mld\bot$, change the latter to $A$  
using Identity(a) perhaps in combination with Commutativity(a).  Next, if $D$ has a surface occurrence of the form $\bot\mlc A$  or $A\mlc\bot$, change it to $\bot$  using Domination(b)   perhaps in combination with Commutativity(b).

{\em Stage 2}: If $D$ has a surface occurrence of the form $(A\mlc B)\mld C$ or $C\mld(A\mlc B)$, change it to $(A\mld C)\mlc (B\mld C)$ using Distributivity(a) perhaps in combination with Commutativity(a). 

{\em Stage 3}: If $D$ has a surface occurrence of the form  $(A\adc^c B)\mld C$ or $C\mld(A\adc^c B)$, change it to $(A\mld C)\adc^c (B \mld C)$ using Distributivity(b) perhaps in combination with Commutativity(a). 


{\em Stage 4}: If $D$ has a surface occurrence of the form $A_1\mld \ldots\mld A_n$ and, for some elementary letter $p$, both 
$p$ and $\neg p$ are among $A_1,\ldots,A_n$, change $A_1\mld \ldots\mld A_n$ to $\top$ using Trivialization, perhaps in combination with Domination(a), Commutativity(a) and Associativity(a).

{\em Stage 5}: If $D$ has a surface occurrence of the form $\top\mld A$ or $A\mld \top$, change it to  $\top$  using Domination(a)   perhaps in combination with Commutativity(a). Next, if $D$ has a surface occurrence of the form
$\top \mlc  A$ or $A\mlc\top$,  change it to $A$  
using Identity(b) perhaps in combination with Commutativity(b).  

{\em Stage 6}: If $D$ has a surface occurrence of the form $(A \adc^a B)\mlc (E \adc^b F)$, change it to 
$\Bigl(\bigl(A\mlc (E \adc^b F)\bigr)\adc^a \bigl(B\mlc (E \adc^b F)\bigr)\Bigr) 
\adc^c \Bigl(\bigl((A\adc^a B)\mlc E\bigr)\adc^b\bigl((A\adc^a B)\mlc D\bigr)\Bigr)$ using Quadrilemma. 


{\em Stage 7}: If $D$ is of the form $X[E\adc^c F]\adc^c A$ (resp. $A\adc^c X[E\adc^c F]$), change it to $X[E]\adc^cA$ (resp. $A\adc^c X[F]$) using Cleansing.

\begin{lemma}\label{terminate}  
Each stage of the Purification procedure strictly reduces the rank of $D$.  
\end{lemma}

\begin{proof} Each stage replaces an occurrence of a subcirquent $A$ of $D$ by some cirquent $B$. In view of Lemma \ref{monot}, in order to show that such a replacement reduces the rank $\overline{D}$ of $D$, it is sufficient to show that $\overline{B}<\overline{A}$.   Keep in mind that the rank of a cirquent is always at least $1$.

{\em Stage 1}: Each iteration of this stage replaces in $D$ an occurrence of $\bot\mld A$,  $A\mld\bot$,  $\bot\mlc A$ or $A\mlc\bot$  by $A$ or $\bot$. Of course, both $\overline{A}$ and $\overline{\bot}$ are smaller than $\overline{\bot\mld A}$,  $\overline{A\mld\bot}$,  $\overline{\bot\mlc A}$ and $\overline{A\mlc\bot}$. 

{\em Stage 2}: Each iteration of this stage replaces  in $D$  an occurrence of  $(A\mlc B)\mld C$ or $C\mld(A\mlc B)$ by  $(A\mld C)\mlc (B\mld C)$. 
$\overline{(A\mlc B)\mld C}$ (or $\overline{C\mld(A\mlc B)}$) is
$ ^{[5^{(\overline{A}+\overline{B})}+\overline{C}]}5$ and $\overline{(A\mld C)\mlc (B\mld C)}$ is  $5^{[^{(\overline{A}+\overline{C})}5+^{(\overline{B}+\overline{C})}5]}$. 
We want to show that $5^{[^{(\overline{A}+\overline{C})}5+^{(\overline{B}+\overline{C})}5]} < \hspace{2pt}^{[5^{\overline{A}+\overline{B}}+\overline{C}]}\hspace{-1pt}5$.  
We of course have $\overline{A}+\overline{B}+1< 5^{\overline{A}+\overline{B}}$, whence $\overline{A}+\overline{B}+\overline{C}+1< 
5^{\overline{A}+\overline{B}}+\overline{C}$, whence 
$^{[\overline{A}+\overline{B}+\overline{C}+1]}5
<^{[5^{\overline{A}+\overline{B}}+\overline{C}]}5$. We also have

\[5^{[^{(\overline{A}+\overline{C})}5+^{(\overline{B}+\overline{C})}5]}=  5^{[^{(\overline{A}+\overline{C})}5]}\times 5^{[^{(\overline{B}+\overline{C})}5]}= ^{(\overline{A}+\overline{C}+1)}5\times ^{(\overline{B}+\overline{C}+1)}5 \leq ^{[\overline{A}+\overline{B}+\overline{C}+1]}5.\] 
 Consequently, $5^{[^{(\overline{A}+\overline{C})}5+^{(\overline{B}+\overline{C})}5]}<^{[5^{\overline{A}+\overline{B}}+\overline{C}]}5$, as desired.

{\em Stage 3}:  $\overline{(A\adc^c B)\mld C}$ (or $\overline{C\mld(A\adc^c B)}$ is $^{(\overline{A}+\overline{B}+\overline{C})}5$, and   $\overline{(A\mld C)\adc^c (B \mld C)}$ is $^{(\overline{A}+\overline{C})}5+^{(\overline{B}+\overline{C})}5$. Taking into account that ranks are always positive, we obviously have  $^{(\overline{A}+\overline{C})}5+^{(\overline{B}+\overline{C})}5<
^{(\overline{A}+\overline{B}+\overline{C})}5$. 

{\em Stage 4}: $\overline{\top}=1<\overline{A_1\mld \ldots\mld A_n}$.

{\em Stage 5}: Similar to Stage 1. 

{\em Stage 6}: $\overline{(A \adc^a B)\mlc (E \adc^b F)}$ is $5^{[\overline{A}+\overline{B}+\overline{E}+\overline{F}]}$, and    
\[\overline{\Bigl(\bigl(A\mlc (E \adc^b F)\bigr)\adc^a \bigl(B\mlc (E \adc^b F)\bigr)\Bigr) 
\adc^c \Bigl(\bigl((A\adc^a B)\mlc E\bigr)\adc^b\bigl((A\adc^a B)\mlc D\bigr)\Bigr)}\]
is $5^{(\overline{A}+\overline{E}+\overline{F})}+5^{(\overline{B}+\overline{E}+\overline{F})}+5^{(\overline{A}+\overline{B}+\overline{E})}+ 5^{(\overline{A}+\overline{B}+\overline{F})}$. Obviously the latter is smaller than the former.

{\em Stage 7}: Each iteration of this stage replaces a subcirquent $E\adc^c F$ by $E$ (resp. $F$). The rank $\overline{E}+\overline{F}$ of $E\adc^c F$ is greater than the rank $\overline{E}$ of $E$ (resp. the rank $\overline{F}$ of $F$).   
 \end{proof}

Where $A$ is the initial value of $D$ in the Purification procedure and $B$ is its final value (which exists by Lemma \ref{terminate}), we call $B$ the {\bf purification} of $A$.

\begin{lemma}\label{pl}
 For any cirquent $A$ and its purification $B$, we have: 

1. If $B$ is provable, then so is $A$. 

2. $A$ is valid iff so is $B$. 

3. $B$ is pure. 

4. The rank of $B$ does not exceed the rank of $A$.  
\end{lemma}

\begin{proof} {\em Clause 1}: When obtaining $B$ from $A$, each transformation performed during the Purification procedure applies, in the conclusion-to-premise sense, one of the inference rules of $\lll$. Reversing the order of those transformations, we get a derivation of $A$ from $B$. Appending that derivation to a proof of $B$ 
(if one exists) yields a proof of $A$. 

{\em Clause 2}: Immediate from the two clauses of Lemma \ref{prescor} and the fact that, when obtaining $B$ from $A$ 
using the Purification procedure, the rule of Choosing is never used. 

{\em Clause 3}: One by one, Stage 1 eliminates all surface occurrences of $\bot$ in $D$ (unless $D$ itself is $\bot$). So, at the end of the stage, $D$ satisfies condition 1 of Definition \ref{pd}. None of the subsequent steps make $D$ violate that condition, so $B$, too, satisfies that condition.
 Similarly, a routine examination of the situation reveals that Stage 2 (resp. 3, \ldots, resp. 7) of the Purification procedure makes $D$ satisfy condition 2 (resp. 3, \ldots, resp. 7) of Definition \ref{pd}, and $D$ continues to satisfy that condition throughout the rest of the stages. So, $B$ is pure.   

{\em Case 4}:  Immediate from Lemma  \ref{terminate}. \end{proof}

\begin{theorem}\label{theo}
A cirquent is valid if (soundness) and only if (completeness) it is provable in \lll.
\end{theorem}

\begin{proof} The soundness part is immediate from clause 1 of Lemma \ref{prescor} and the fact that the axiom $\top$ is valid. The rest of this section is devoted to a proof of the completeness part. Pick an arbitrary cirquent $A$ and assume it is valid. We proceed by induction on the 
rank of $A$. Let $B$ be the purification of $A$.

 In view of clauses 2-4 of Lemma \ref{pl}, $B$ is a valid, pure  cirquent whose rank does not exceed that of 
$A$. We shall implicitly rely on this fact below. By clause 1 of Lemma \ref{pl}, if $B$ is provable, then so is $A$. Hence, in order to show that $A$ is provable, it suffices to show that $B$ is provable. $B$ cannot be $\bot$ because then, of course, it would not be valid. Similarly, $B$ 
cannot be a literal because obviously no literal is valid. In view of this observation and $B$'s being pure,  it is clear that  the following  cases 
cover all possibilities for $B$. 

{\em Case 1}:  $B$ is $\top$. Then $B$ is an axiom and hence provable. 

{\em Case 2}: $B$ is $E \add^c  F$. Let $\cal H$ be a logical solution of $B$. Consider the work of $\cal H$ in the scenario where the environment does not move until $\cal H$ makes the move $c.i$, where $i\in\{0,1\}$. Sooner or later $\cal H$ has to make such a move, for otherwise $B$ would be lost due to being $\add^c$-rooted.  Since in the games that we deal with the order of moves is irrelevant, without loss of  generality we may assume that the move $c.i$ is made before any other moves. Let $B'$ be the result of replacing in $B$ all subcirquents of the form $X_0\add^c X_1$ by $X_i$.  Observe that, after the move $c.i$ is made, in any scenario that may follow, $\cal H$ has to continue and win $B'$. In other words,  $\cal H$ is a logical solution of (not only $B$ but also) $B'$. The rank of $B'$ is of course smaller than that of $B$. Hence, by the induction hypothesis, $B'$  is provable. Then $B$ follows from $B'$ by Choosing. 

{\em Case 3}: $B$ is $E \adc^c F$, and neither $E$ nor $F$ contains the cluster $c$. By clause 2 of Lemma \ref{prescor}, both $E$ and $F$ are valid, because $B$ follows from them by Splitting. The rank of either cirquent is smaller  than that of $B$. Hence, by the induction hypothesis, both $E$ and $F$ are provable. Therefore, by Splitting, so is $B$.

{\em Case 4}: $B$ is $E_1\mld\ldots\mld E_n$ ($n\geq 2$), where each $E_i$ is either a literal of a cirquent of the form $F \add^c G$; besides, for no elementary letter $p$ do we have that both $p$ and $\neg p$ are among $E_1,\ldots,E_n$. Not all of the cirquents $E_1,\ldots,E_n$ can be literals, for otherwise $B$ would be automatically lost under an interpretation which interprets all those literals as $\bot$, contradicting our assumption that $B$  is valid. With this observation in mind, without loss of generality, we may assume that, for some 
$k$ with $1\leq k \leq n$, the first $k$ cirquents $E_1,\ldots,E_k$ are of the form $F_1\add^{c_1} G_1$, \ldots, 
$F_k\add^{c_k} G_k$ and the remaining $n-k$ cirquents $E_{k+1},\ldots,E_n$ are literals. Let $\cal H$ be a 
logical solution of $B$. Consider the work of $\cal H$ in the scenario where the environment makes no moves. 
Note that, at some point, for some $1 \leq j \leq k$, $\cal H$ should make the move $c_j.i$ ($i\in\{0.1\}$), 
for otherwise $B$ 
would be lost under an(y) interpretation which interprets all of the literal cirquents $E_{k+1},\ldots,E_n$ as $\bot$. Fix such  $j,i$. Let $B'$ be the result of replacing, in $B$, every subcirquent of the form $X_0 \add^{c_j} X_1$  by $X_i$.  With some analysis left to the reader, $\cal H$ can be seen to be a logical solution of $B'$. Thus, $B'$ is valid. The rank of $B'$ is smaller than that of $B$ and hence, by the induction hypothesis, $B'$ is provable. But then so is $B$, because it follows from $B'$ by Choosing. 

{\em Case 5}: $B$ is $E_1\mlc\ldots\mlc E_n$ ($n\geq 2$), where, for some $e$ ($1\leq e \leq n$), $E_e$ --- fix it --- is not of the form $F\adc^c G$ or $F\mlc G$, nor do we have $E_e\in\{\top,\bot\}$. The validity of $B$, of course, implies that 
$E_e$, as one of its $\mlc$-conjuncts, is also valid. This rules out the possibility that $E_e$ is a literal, because, as we observed earlier, a literal cannot be valid. We are therefore left with one of the following two possible subcases:

{\em Subcase 5.1}: $E_e$ is of the form $F\add^c G$. Let $\cal H$ be a logical solution of $B$. As in Case 4, consider the work of $\cal H$ in the scenario where the environment makes no moves. Note that, at some point, $\cal H$ should make the move $c.0$ or $c.1$, for otherwise $B$ would be lost (under any interpretation). Let us just consider the case of the above move being $c.0$ (the case of it being $c.1$ will be handled in a similar way). Let $B_0$ be the result of replacing, in $B$, every subcirquent of the form $X \add^c Y$ 
(including the conjunct $F\add^c G$) by $X$. Then, as in  Case 4, $\cal H$ can be seen to be a logical solution of $B_0$. Thus, $B_0$ is valid. The rank of $B_0$ is smaller than that of $B$ and hence, by the induction hypothesis, $B_0$ is provable. But then so is $B$, because it follows from $B_0$ by Choosing(a). 

{\em Subcase 5.2}: $E_e$ is of the form $F_1\mld\ldots\mld F_m$, where each $F_i$ ($1\leq i\leq m$) is either a literal or a cirquent of the form $G\add^c H$, and for no elementary letter $p$ do we have that both $p$ and $\neg p$ are among $F_1,\ldots,F_m$. This 
case is very similar to Case 4 and, almost literally repeating our reasoning in the latter, we find that $B$ 
is provable. \end{proof}

\end{document}